\documentclass[12pt]{article}
\usepackage{amssymb, amsmath,framed,transparent,enumerate,color}
\usepackage[colorlinks]{hyperref}
\usepackage[margin=2cm]{geometry}

\newcommand{\rem}[1]{}
\newcommand{\de}{{\rm d}}

\newcommand{\bx}{{\mathbf{x}}}
\newcommand{\bm}{{\mathbf{m}}}

\newcommand{\bF}{{\mathbf{F}}}

\newcommand{\bB}{{\mathbf{B}}}

\newcommand{\bu}{{\mathbf{u}}}

\newcommand{\bel}{\begin{equation}}
\newcommand{\ee}{\end{equation}}
\newcommand{\beq}{\begin{eqnarray}}
\newcommand{\eq}{\end{eqnarray}}
\newcommand{\eeq}{\end{eqnarray}}
\newcommand{\bit}{\begin{itemize}}
\newcommand{\eit}{\end{itemize}}
\newcommand{\ben}{\begin{enumerate}}
\newcommand{\een}{\end{enumerate}}
\newcommand{\bc}{\begin{center}}
\newcommand{\ec}{\end{center}}

\definecolor{Orange}{rgb}{1.000, .529, .000}
\definecolor{AntiqueWhite}{rgb}{ .804, .753, .690}
\definecolor{Violet}{cmyk}{0.79,0.88,0,0}
\definecolor{Plum}{cmyk}{0.50,1,0,0}
\definecolor{Periwinkle}{cmyk}{0.57,0.55,0,0}
\definecolor{ForestGreen}{cmyk}{0.91,0,0.88,0.12}
\definecolor{OliveGreen}{cmyk}{0.64,0,0.95,0.40}
\definecolor{BrickRed}{cmyk}{0,0.89,0.94,0.28}
\definecolor{DarkOrchid}{cmyk}{0.40,0.80,0.20,0}
\definecolor{Fuchsia}{cmyk}{0.47,0.91,0,0.08}
\definecolor{Mulberry}{cmyk}{0.34,0.90,0,0.02}
\definecolor{Maroon}{cmyk}{0,0.87,0.68,0.32}
\definecolor{Mahogany}{cmyk}{0,0.85,0.87,0.35}
\definecolor{RawSienna}{cmyk}{0,0.72,1,0.45}
\definecolor{YellowOrange}{cmyk}{0,0.42,1,0}
\definecolor{OrangeOrange}{cmyk}{0,0.62,1,0}
\definecolor{BurntOrange}{cmyk}{0,0.51,1,0}
\definecolor{Bittersweet}{cmyk}{0,0.75,1,0.24}
\definecolor{RedOrange}{cmyk}{0,0.77,0.87,0}
\definecolor{Sepia}{cmyk}{0,0.83,1,0.70}
\definecolor{Brown}{cmyk}{0,0.81,1,0.60}
\definecolor{Tan}{cmyk}{0.14,0.42,0.56,0}
\definecolor{darkorange}{cmyk}{.20,.50,.80,0}    
\definecolor{lightorange}{cmyk}{.07,.37,.65,0}   
\definecolor{darkpeagreen}{cmyk}{.50,.30,.50,0}  
\newcommand{\pone}{\par\vspace{.25cm}\noindent}

\newtheorem{theorem}{Theorem}

\newtheorem{lemma}[theorem]{Lemma}
\newtheorem{definition}{Definition}
\newtheorem{remark}{Remark}
\newenvironment{proof}[1][Proof]{\noindent\textbf{#1.} }{\ \rule{0.5em}{0.5em}}

\begin{document}

\title{
Multiscale Turbulence Models\\ Based on Convected Fluid Microstructure
}
\author{Darryl D Holm
\\ Mathematics Department
\\ Imperial College London
\\  and \\ 
Cesare Tronci
\\ Mathematics Department
\\ Surrey University}

\date{\it In honor of Peter Constantin's 60th birthday.}

\maketitle

\begin{abstract}
The Euler-Poincar\'e approach to complex fluids is used to derive multiscale equations for computationally modelling Euler flows as a basis for modelling turbulence. The model is  based on a \emph{kinematic sweeping ansatz} (KSA) which assumes that the mean fluid flow serves as a Lagrangian frame of motion for the fluctuation dynamics. Thus, we regard the motion of a fluid parcel on the computationally resolvable length scales as a moving Lagrange coordinate for the fluctuating (zero-mean) motion of fluid parcels at the unresolved scales. Even in the simplest 2-scale version on which we concentrate here, the contributions of the fluctuating motion under the KSA to the mean motion yields a system of equations that extends known results and appears to be suitable for modelling nonlinear backscatter (energy transfer from smaller to larger scales) in turbulence using multiscale methods.
\end{abstract}

\tableofcontents

\section{Multiscale approaches}

\subsection{Dealing with microstructure dynamics in turbulence}
The history and present state of the art of multiscale approaches to fluid turbulence modelling are recounted in a number of excellent sources, including \cite{EE03,Hou2005,PS2008,EfendHou2009,Tr2010}. For the present purpose of  modelling multiscale fluid turbulence, we mention Multiscale Finite Elements Methods \cite{EfendHou2009} and the Heterogeneous Multiscale Method \cite{E-etal}. 
Both are general methodologies for numerical computation of problems with multiple scales. The methods rely on an efficient coupling between the different macroscopic and microscopic physical models. The key to the efficiency of such an approach is the possibility that the microscale model may not need to be solved over the entire computational domain, but only over small selected regions near where data estimation is carried out. Examples of applications include complex fluids, micro-fluidics, solids, interface problems, stochastic problems, and statistically self-similar problems.

The present work applies the standard multiscale method to Euler's fluid equations, then combines the results with ideas from the geometric mechanics of complex fluids, in order to create a new two-scale model of ideal incompressible flow. This is accomplished by: (i) making a slow-fast spatial decomposition of the fluid velocity; (ii) performing Lagrangian averaging in Hamilton's principle using this decomposition; and then (iii) applying the Taylor hypothesis in assuming that fluctuations are convected by the mean flow. The last step treats the flow trajectories of the mean flow as Lagrangian coordinates for the dynamics of the fluid fluctuations.
The result is a two-scale extension of the Lagrangian-averaged Euler alpha equations of \cite{HoMaRa1998} which was the basis for the Lagrangian-averaged Navier-Stokes alpha (LANS-$\alpha$) equations of \cite{FoHoTi2001,FoHoTi2002}.

\subsection{Convection of microstructure}
Homogenization techniques were applied in \cite{MPP1985} to obtain an averaged equation for the large scale features of highly-oscillatory  solutions of the three-dimensional incompressible Euler or Navier-Stokes equations. The following initial value problem was treated:
\[
\partial_tu+ (u \cdot \nabla)u = -\nabla p,
\quad\hbox{with}\quad 
\nabla \cdot u = 0 
\]
and with highly-oscillatory initial data ($\epsilon\ll1$)
\[
u(x, 0) = U(x) +W\left(x, \frac{x}{\epsilon}\right).
\]
Multiscale expansions were constructed for both the velocity field and the pressure, under the important assumption that the \emph{microstructure is advected by the mean flow}. Under this assumption, the following multiscale expansion for the velocity field was constructed:
\[
u^\epsilon(x, t) = u(x, t) + w\left(\frac{\theta(x,t)}{\epsilon},\frac{t}{\epsilon},\frac{x}{\epsilon},x,t\right)
+ \epsilon  u_1\left(\frac{\theta(x,t)}{\epsilon},\frac{t}{\epsilon},\frac{x}{\epsilon},x,t\right)
+O\left(\epsilon^2\right)
.
\]
The pressure field $p^\epsilon$ was expanded similarly. This form of the solutions for the velocity and pressure fields were shown to be consistent with the fluctuation quantity $\theta$ being advected by the \emph{mean velocity} as a Lagrangian coordinate. Namely, 
\begin{equation}\label{Adv-microstructure}
\partial_t\theta + u\cdot \nabla \theta = 0
\,,\quad
\theta(x,0) = x
\,.
\end{equation}
The additional vector variable $\theta$ is the back-to-labels map, or inverse map for the three-dimensional incompressible Euler equations.  Modelling the effects of the rapid small scales on the slower large-scale solutions of the three-dimensional Euler and Navier-Stokes equations constitutes the closure problem in turbulence theory. The work in \cite{MPP1985}  provided some understanding of the interactions of the small scales with the large scales and it characterized  the back-to-labels map as a form of fluctuating microstructure attached to the mean flow. The closure problem for turbulence was not solved, however, because the solutions for the functions $u$ and $w$ of the rapidly oscillating variables turned out not to be unique. The uniqueness problem was addressed by imposing additional assumed  constraints that led to large-scale averaged equations that resembled the then-popular $k$-$\epsilon$ closure model of turbulence.

Convection of microstructure of the two- and three-dimensional incompressible Euler equations has also been studied from a related but different viewpoint from \cite{MPP1985} in a series of recent papers, culminating in \cite{HYR2008}. These papers used a new approach that does not assume that the fluctuations follow the characteristics of the mean flow. This new approach enabled the derivation of a well-posed cell problem which was used to obtain an effective large-scale average equation. However, the numerical solution still required solving for the back-to-labels map at each time step.  

The present work builds on the results of this previous research, and obtains systems of multiscale equations that are purely Eulerian and avoid the necessity of solving for the back-to-labels map. The approach is based on the Euler-Poincar\'e  theory of complex fluids developed during the past decade, see e.g. \cite{HT2009} and references therein, also in a series of other papers \cite{GBTRPH}. The complex fluids approach is based on geometric mechanics and on the application of reduction by symmetry to Hamilton's principle for ideal continua \cite{HMR1998}. The corresponding Hamiltonian approach to complex fluids was developed earlier by Dzyaloshinskii and Volovik in their famous paper \cite{DzVo1980}. Later, a Hamilton's principle approach was followed to establish the variational approach to various soft matter systems in \cite{GBTRPH}, where various complex fluids were studied in different contexts, including dissipative dynamics.

\subsection*{Plan of the paper and main results}
\begin{enumerate} [(i)]
\item
Section 2 discusses the two-scale fluid flow decomposition that we use and introduces Taylor's hypothesis as a kinematic sweeping ansatz (KSA).
\item
Section 3 formulates the two-scale model by applying geometric methods for Hamilton's principle that were borrowed from complex fluids theory. 
\item
Section 4 discusses conservation laws for total subgrid scale circulation and helicity.
\item
Section 5 develops resolved-scale models based on advection laws for subgrid scale quantities. These models are similar to the stretched subgrid scale vortex model of \cite{MP1997}.
\item
Section 6 discusses a subgrid scale fluid model with feedback between the two scales.
\item
Section 7 contains a summary and conclusions of the paper. 
\item
There are also five Appendices that explain details of proofs and derivations sketched in the body of the text. 
\end{enumerate}

\section{Flow decomposition and Kinematic Sweeping Ansatz}

\subsection{Decomposition of periodic vector functions}

\begin{lemma}[Decomposition Lemma \cite{HYR2008}]\label{decomp-Lemma}$\,$

Any periodic function $\mathbf{u}(\mathbf{x}):\mathbb{R}^3\to\mathbb{R}^3$ that admits a Fourier series expansion on the unit cube $[0, 1]^3$ may be decomposed into a sum of two periodic functions
\begin{align}
\mathbf{u}(\mathbf{x}) = \mathbf{u}_1(\mathbf{x}_1) + \mathbf{u}_2(\mathbf{x}_1,\mathbf{x}_2)
\quad\hbox{with}\quad
\mathbf{x}_1,\mathbf{x}_2\in[0, 1]^3
\label{decomp-lemma}
\end{align}
in which $\mathbf{u}_2$ has zero mean in $\mathbf{x}_2\in[0, 1]^3$. That is, $\int  \mathbf{u}_2(\mathbf{x}_1,\mathbf{x}_2)\,\de^3x_2=0$.
\end{lemma}

The proof follows from manipulation of Fourier series and is provided in the present notation in Appendix \ref{lem1-proof}.
We  will refer to $\mathbf{u}_1$ and $\mathbf{u}_2$ in equation (\ref{decomp-lemma}) as the velocities of the mean flow and the fluctuations, respectively. The $\mathbf{x}_1$ dependence in $\mathbf{u}_2(\mathbf{x}_1,\mathbf{x}_2)$ may be regarded as the slowly varying envelope of the rapidly fluctuating component of the velocity. The fluid pressure has a similar decomposition. 
\bigskip

\paragraph{Potential for applications in numerics.}
We will use Lemma \ref{decomp-Lemma} to express Euler's fluid equations in terms of two scales $(\mathbf{x}_1,\mathbf{x}_2)=(\mathbf{x},\mathbf{x}/\epsilon)$, in which we will regard $\epsilon\in[0,1]$ as a fixed parameter. In a computational setting, for example, the value of $\epsilon$ could be assigned by the choice of grid size used to resolve the large-scale solution $u_1(\mathbf{x})$, while regarding the remainder $u_{2}(\mathbf{x},\mathbf{x}/\epsilon)$ as the unresolved, subgrid-scale part of the solution. Thus, $\epsilon=1/N$ in this case corresponds to the fractional length scale that one would like to resolve on a computational grid. The limit $\epsilon\to0$ would pick up more and more small-scale components of the solution, and $u_{2}(\mathbf{x},\mathbf{x}/\epsilon)$ would tend to zero in the limit. That is, letting $\epsilon\to0$ resolves more and more small scales into the solution. In that limit, the computation would become a direct numerical simulation which attempts to resolve all scales. However, that limit is not our objective. 

The challenge for us here is to introduce a computable fluid model that describes the effect of subgrid scales on the resolvable scales $\mathbf{x}_1$ at a fixed nonzero value of $\epsilon$. Our approach will be to use the hypothesis of \emph{convected fluid microstructure} \cite{MPP1985,HYR2008}. For this, we will regard the motion  $\mathbf{x}_1(t)$ of a fluid parcel on the resolvable length scale as a moving Lagrange coordinate for the motion $\mathbf{x}_2(\mathbf{x}_1(t),t)$ of fluid parcels at the unresolved scales. We will apply this assumption by using methods of geometric mechanics. After some kinematic considerations for convection of fluid microstructure, the result will be an Eulerian multiscale fluid model. 

\subsection{Lagrangian averaging: fluctuations are swept by the mean}

The main hypothesis of the present paper is that fluctuations are swept by the mean flow, but they are not slaved to the mean flow as in  Large Eddy Simulation (LES) turbulence modelling. In particular, we assume that the fluctuations are swept by the large-scale motion and they have zero mean \emph{in the Lagrangian frame moving with the large-scale velocity}. 

Our interpretation of fluctuations is that they undergo their own evolution, which in turn depends on the Lagrangian fluid parcel traveling with the mean flow. That is, at each mean Lagrangian label $\boldsymbol\psi(\mathbf{x}_1^{(0)})$, there corresponds another Lagrangian label $\boldsymbol\eta_{\mathbf{x}_1}(\mathbf{x}_2^{(0)})$  that is associated to fluctuation dynamics. More specifically, we introduce the following definition.

\begin{definition}[KSA: Kinematic Sweeping Ansatz]\label{KSA}$\,$

KSA: Fluctuations are swept by the mean. 
That is, while the Lagrangian label $\mathbf{x}_1^{(0)}$ for the mean flow is taken to its current position by the map
\[
\mathbf{x}_1=\boldsymbol\psi(\mathbf{x}_1^{(0)})
\,,
\]
the fluctuation label $\mathbf{x}_2^{(0)}$ transforms according to
\[
\mathbf{x}_2=\boldsymbol\eta_{\mathbf{x}_1}(\mathbf{x}_2^{(0)})
\,.
\]
The subscript notation in the above formula emphasizes that the Lagrangian map $\boldsymbol\eta$ for the fluctuations depends on the Eulerian position $\mathbf{x}_1$ of the mean flow parcel. The maps $\boldsymbol\psi$ and $\boldsymbol\eta$ also both depend explicitly on time. Moreover, both the mean flow and fluctuations are assumed to undergo incompressible dynamics, so the maps $\boldsymbol\psi$ and $\boldsymbol\eta$ each preserve their respective volume elements
\[
\de^3 {x}_1^{(0)}=\de^3{x}_1
\qquad\text{and}\qquad
\de^3 {x}_2^{(0)}=\de^3{x}_2
\,.\]
\end{definition}\bigskip

As a result of the KSA, the velocity of a fluctuation as it is swept along a Lagrangian mean trajectory $\mathbf{x}_1(t)$ may be written in the  form
\begin{equation}\label{Fluctuations}
\frac{\de\mathbf{x}_2}{\de t}
=
 \frac{\partial}{\partial t}\boldsymbol\eta_{\mathbf{x}_1}(\mathbf{x}_2^{(0)})
 +
 \mathbf{u}_1\cdot\nabla_1\boldsymbol\eta_{\mathbf{x}_1}(\mathbf{x}_2^{(0)})
\,,
\end{equation}
where $\mathbf{u}_1(\mathbf{x}_1)=\boldsymbol{\dot{\psi}}(\boldsymbol\psi^{-1}(\mathbf{x}_1))$ is the Eulerian mean velocity and $\nabla_1$ stands for $\partial/\partial{\bf x}_1$. 

\begin{remark}\rm
The main difference between the present two-scale sweeping approach and other approaches in the literature lies in the second material term $\mathbf{u}_1(\mathbf{x}_1)\cdot\nabla_1\boldsymbol\eta_{\mathbf{x}_1}$, which shows that the fluctuations are dragged by the mean flow in a Lagrangian sense, so that Lagrangian mean trajectories $\mathbf{x}_1(t)$ become Lagrangian variables for the fluctuation dynamics. 
\end{remark}

The fluctuations will be constrained to have zero mean. However, the question arises of how the mean should be taken. It is clear that the precise quantity possessing zero mean is the Eulerian fluctuation velocity $\mathbf{u}_2(\mathbf{x}_1,\mathbf{x}_2,t)$. However, one must ask in which sense the latter has constant zero mean. For example, under purely Eulerian reasoning, one might be tempted to set $\partial_t\!\int\!\mathbf{u}_2(\mathbf{x}_1,\mathbf{x}_2,t)\,\de^3\mathbf{x}_2=0$ so that the mean of the fluctuations would remain constant in time at a given point $\mathbf{x}_1$. On the other hand, if the fluctuations are swept along by the mean flow, then one must write the Eulerian version of the Lagrangian time derivative as
\begin{align}
\left(\frac{\partial}{\partial t}+\mathbf{u}_1\cdot\nabla_1\right)
\!\int\!\mathbf{u}_2(\mathbf{x}_1,\mathbf{x}_2,t)\,\de^3{x}_2=0
\,,
\label{mean-preserv}
\end{align}
so that the mean of the fluctuation velocity remains constant \emph{along the Lagrangian mean paths of the $\mathbf{u}_1$-flow}. The above relation is the second main ingredient of our approach. While the basic assumption of a mean flow sweeping fluctuations in a Lagrangian sense is taken as an underlying hypothesis, the transport equation for the mean of the fluctuation will arise as a consequence of our treatment.

\subsection{The kinematic model}

Applying Lemma \ref{decomp-Lemma} and the KSA directly to Euler's fluid equation, and setting terms in $\mathbf{x}_2$ separately to zero yields the following equations of motion for $\mathbf{u}_1$ and $\mathbf{u}_2$,
\begin{align}
&\frac{\partial}{\partial t}\mathbf{u}_1 + ( \mathbf{u}_1\cdot\nabla_1)\mathbf{u}_1
=
-
\nabla_1\mathsf{p}_1(\mathbf{x}_1)
\,,\qquad
\nabla_1\cdot\mathbf{u}_1 =0\,,
\label{BigWhirls-mot}\\
&
\frac{\partial}{\partial t}\mathbf{u}_2 + ( \mathbf{u}_1\cdot\nabla_1)\mathbf{u}_2
+
 ( \mathbf{u}_2\cdot\nabla_2)\mathbf{u}_2
=
-\nabla_2\mathsf{p}_2(\mathbf{x}_1,\mathbf{x}_2)
\,,\qquad
\nabla_2\cdot\mathbf{u}_2 =0\,,
\label{SmallWhirls-mot}
\end{align}
in which the pressures $p_1$ and $p_2$ are determined by preservation of incompressibility of the velocities $\mathbf{u}_2$ and $\mathbf{u}_2$, respectively.
As we shall see later, the above equations also follow naturally by applying Lemma \ref{decomp-Lemma} and the KSA to Hamilton's principle.
The first equation (\ref{BigWhirls-mot}) yields Euler's fluid equation for the mean velocity $\mathbf{u}_1$, while the fluctuation velocity $\mathbf{u}_2$ is transported by the term $(\mathbf{u}_1\cdot\nabla_1){\mathbf{u}_2}$ and also undergoes its own nonlinear Euler dynamics, given by the term $( \mathbf{u}_2\cdot\nabla_2)\mathbf{u}_2$ in the second equation (\ref{SmallWhirls-mot}). 

Equation (\ref{SmallWhirls-mot}) preserves the mean obtained by integration over $\de^3x_2$. Moreover,  integration of equation (\ref{SmallWhirls-mot}) over $\de^3x_1$ produces an averaged version of Euler's equation in the form with $\langle\,\cdot\,\rangle_1 = \int (\,\cdot\,)\,\de^3x_1$
\[
\frac{\partial}{\partial t}\langle\mathbf{u}_2\rangle_1 + \nabla_2\,\cdot \langle\mathbf{u}_2\otimes\mathbf{u}_2\rangle_1 = -\nabla_2\langle\,p_2\,\rangle_1
\,.
\]
Consequently, $\frac{d}{dt}\int\langle\mathbf{u}_2\rangle_1\,\de^3x_2=0$, which is the statement of preservation of the zero-mean condition for the Eulerian fluctuation velocity $\mathbf{u}_2(\mathbf{x}_1,\mathbf{x}_2,t)$.

\paragraph{Kinetic energy conservation.}
The total kinetic energy of a fluid flow is the $L^2$ norm of its velocity $\|\mathbf{u}\|^2_{L^2}$ on the domain of flow, which decomposes according to Lemma \ref{decomp-Lemma} into 
\begin{equation}
\|\mathbf{u}\|^2_{L^2} 
= \int |\mathbf{u}_1(\mathbf{x})+\mathbf{u}_2(\mathbf{x},\mathbf{x}/\epsilon)|^2 \de^3x
\,.
\label{KinKE}
\end{equation}
One takes the average over the rapid variations by integrating over $\mathbf{x}/\epsilon=:\mathbf{x}_2$ while holding  $\mathbf{x}=:\mathbf{x}_1$ fixed. This yields, 
\begin{equation}
\left\langle\|\mathbf{u}\|^2_{L^2} \right\rangle_2
= \int |\mathbf{u}_1(\mathbf{x}_1)|^2\de^3x_1
+ \frac{1}{V_2}\int |\mathbf{u}_2(\mathbf{x}_1,\mathbf{x}_2)|^2\de^3x_1\, \de^3x_2
=
\|\mathbf{u}_1\|^2_{L^2} + \left\langle\|\mathbf{u}_2\|^2_{L^2} \right\rangle_2
\,,
\label{KinKEavg}
\end{equation}
where $V_2=1$ is the volume of the domain $\mathbf{x}_2\in[0, 1]^3$ and the zero-mean relation $\int  \mathbf{u}_2(\mathbf{x}_1,\mathbf{x}_2)\,\de^3x_2=0$ has been used.
Thus, the mean total fluid kinetic energy decomposes into the sum of the square of the $L^2$ metric of the velocity $\mathbf{u}_1$ and the mean-square $L^2$ metric of the velocity $\mathbf{u}_2$. 

\begin{remark}\rm
The two kinetic energy norms in (\ref{KinKE}) are conserved \emph{separately} by the system (\ref{BigWhirls-mot})--(\ref{SmallWhirls-mot}).
\end{remark}

\begin{remark}[Relation to the alpha-models]\rm$\,$\\
The treatment so far mimics the treatment of Lagrangian averages of the WKB decomposition in \cite{GjHo1996}. That  approach led to the Lagrangian Averaged Navier-Stokes alpha model of turbulence \cite{CFHOTW1998,CFHOTW1999} in which Taylor's hypothesis \cite{Ho2005} was invoked as a closure, by imposing that small excitations $k\alpha>1$ evolve by being swept by the larger scales $k\alpha<1$, under which the nonlinearity of the smaller scales is ignored. Here, that assumption has been \emph{relaxed} as in equation (\ref{SmallWhirls-mot}), to allow for the smaller scales to evolve under their own nonlinearity, \emph{relative to} the motion of the larger scales, whose flow trajectories are treated as Lagrangian coordinates for the smaller scales. In particular, this means that averaging by integrating over $\mathbf{x}_2$  while holding $\mathbf{x}_1$ fixed may be viewed as \emph{Lagrangian averaging} in this situation. 
\end{remark}

The decoupled equations (\ref{BigWhirls-mot})--(\ref{SmallWhirls-mot}) comprise a simple non-interaction representation of two-scale Euler equations. Namely, equation (\ref{BigWhirls-mot}) has reduced to Euler's equation for the velocity $\mathbf{u}_1$ of the resolved scale motion (big whirls). And the velocity $\mathbf{u}_2$ for the subgrid scale motion (little whirls) evolving in equation (\ref{SmallWhirls-mot}) is governed by Euler's fluid equation, expressed in the moving frame of the $\mathbf{u}_1$-flow, viewed as a scalar transformation applied to the $\mathbf{x}_1$-dependence of the fluctuation velocity $\mathbf{u}_2$. To describe this situation, we say that the fluctuations are \emph{swept, not slaved} by the large scales. However, the KSA would not be enough of a nonlinear basis to describe turbulence, because there is not yet any back-reaction from the fluctuations to the mean motion. Having set up this kinematic framework, the remainder of the paper deals with modelling further dynamical interactions between the mean flow and the fluctuations.

\section{Formulation of the two-scale model}

\subsection{Hamilton's principle for two scales of motion}
We shall derive a two-scale model of turbulence dynamics by using the relabeling symmetry in Hamilton's principle for ideal fluids
\[
\delta\int_{t_1}^{t_2}L(\boldsymbol\psi,\boldsymbol{\dot{\psi}},\boldsymbol\eta,\boldsymbol{\dot{\eta}})\,\de t=0
\]
where we shall keep in mind the $\left\langle\|\mathbf{u}\|^2_{L^2} \right\rangle_2$ kinetic-energy form in equation (\ref{KinKEavg}),
\[
L(\boldsymbol\psi,\boldsymbol{\dot{\psi}},\boldsymbol\eta,\boldsymbol{\dot{\eta}})=\frac12\int\left|\boldsymbol{\dot{\psi}}(\mathbf{x}_1^{(0)})\right|^2 \de^3 {x}_1^{(0)}+\frac12\iint \left|\boldsymbol{\dot{\eta}}_{\boldsymbol\psi(\mathbf{x}_1^{(0)})}(\mathbf{x}_2^{(0)})\right|^2\de^3 x_2^{(0)}\,\de^3 {x}_1^{(0)}.
\]
Upon using the relabeling symmetry, this particular Lagrangian becomes
\begin{align*}
L(\boldsymbol\psi,\boldsymbol{\dot{\psi}},\boldsymbol\eta,\boldsymbol{\dot{\eta}})=\,&\tilde{\ell}(\boldsymbol\psi,\boldsymbol{\dot{\psi}},\boldsymbol{\dot{\eta}}\circ{\boldsymbol\eta}^{-1})
\\
=\,&\frac12\int\left|\boldsymbol{\dot{\psi}}(\mathbf{x}_1^{(0)})\right|^2 \de^3 \mathbf{x}_1^{(0)}+
\frac12\iint \left|\mathbf{\tilde{u}}_2\big(\boldsymbol\psi(\mathbf{x}_1^{(0)}),\mathbf{x}_2\big)\right|^2\de^3{x}_2\,\de^3 {x}_1^{(0)}
\end{align*}
where $\mathbf{\tilde{u}}_2(\boldsymbol\psi(\mathbf{x}_1^{(0)}),\mathbf{x}_2):=\boldsymbol{\dot{\eta}}_{\boldsymbol\psi(\mathbf{x}_1^{(0)})}\circ{\boldsymbol\eta}_{\boldsymbol\psi(\mathbf{x}_1^{(0)})}^{-1}(\mathbf{x}_2)$.
One then finds that
\begin{align}
\tilde{\ell}(\boldsymbol\psi,\boldsymbol{\dot{\psi}},\mathbf{\tilde{u}}_2)&
=
{\ell}(\boldsymbol{\dot{\psi}}\circ\boldsymbol\psi^{-1},\mathbf{\tilde{u}}_2\circ\boldsymbol\psi)
\nonumber
\\
&=
\frac12\int\left|\mathbf{u}_1(\mathbf{x}_1)\right|^2 \de^3 \mathbf{x}_1+
\frac12\iint \left|{\mathbf{u}}_2\big(\mathbf{x}_1,\mathbf{x}_2\big)\right|^2\de^3{x}_2\,\de^3{x}_1
\label{L2-lagrangian}
\end{align}
as in (\ref{KinKEavg}), where $\mathbf{u}_1(\mathbf{x}_1)=\boldsymbol{\dot{\psi}}\circ\boldsymbol\psi^{-1}(\mathbf{x}_1)$ and $\mathbf{u}_2(\mathbf{x}_1,\mathbf{x}_2)=\mathbf{\tilde{u}}_2(\boldsymbol\psi(\mathbf{x}_1),\mathbf{x}_2)$.

At this point, one may take variations and apply Hamilton's principle in its general form
\[
\delta\int_{t_1}^{t_2}\ell(\mathbf{u}_1,\mathbf{u}_2)\,\de t=0\,,
\]
for an arbitrary Lagrangian, $\ell(\mathbf{u}_1,\mathbf{u}_2)$ arising from relabeling symmetry arguments as above.
The variations ${\delta{\mathbf{u}}_1}$ and ${\delta{\mathbf{u}}_2}$ are computed in Appendix \ref{velocity-vars}  and they produce the following dynamics:
\begin{align}
&\frac{\partial}{\partial t}\frac{\delta\ell}{\delta{\mathbf{u}}_1}
+
\mathbf{u}_1\cdot\nabla_1\frac{\delta\ell}{\delta{\mathbf{u}}_1}
+
(\nabla_1\mathbf{u}_1)^T\cdot\frac{\delta\ell}{\delta{\mathbf{u}}_1}
+
\int\! \left(\nabla_1\mathbf{u}_2\right)^{\rm T}\cdot \frac{\delta\ell}{\delta{\mathbf{u}}_2}\,\de^3x_2
=
-
\nabla_1\mathsf{p}_1
\label{M1-eqn}
\\
&\frac{\partial}{\partial t}\frac{\delta\ell}{\delta{\mathbf{u}}_2}
+
\frac{\partial}{\partial x_1^j} \left(u_1^j \frac{\delta\ell}{\delta{\mathbf{u}}_2}\right)
+
\mathbf{u}_2\cdot\nabla_2\frac{\delta\ell}{\delta{\mathbf{u}}_2}
+
(\nabla_2\mathbf{u}_2)^T\cdot\frac{\delta\ell}{\delta{\mathbf{u}}_2}
=
-\nabla_2\mathsf{p}_2
\label{M2-eqn}
\end{align}
where we denote $(\nabla\mathbf{u})^T\cdot\mathbf{v}=v_j\nabla u^j$ for a vector $\mathbf{u}$ and a co-vector $\mathbf{v}$. Equivalently, upon defining 
\begin{align*}
\mathbf{m}_1(x_1):=\frac{\delta\ell}{\delta{\mathbf{u}}_1}
,\quad 
\mathbf{m}_2(x_1,x_2):=\frac{\delta\ell}{\delta{\mathbf{u}}_2}
,\quad
\frac{D}{D t_1}:=\frac{\partial}{\partial t} +  \mathbf{u}_1\cdot\nabla_1
\,,\end{align*}
we may rewrite equations (\ref{M1-eqn})--(\ref{M2-eqn}) equivalently as
\begin{align}
&\frac{D}{D t_1}\mathbf{m}_1
+
 \left(\nabla_1\mathbf{u}_1\right)^{\rm T}\cdot\mathbf{m}_1
+
\underbrace{\
\int\! \left(\nabla_1\mathbf{u}_2\right)^{\rm T}\cdot \mathbf{m}_2\,\de^3x_2}
_{\hbox{\it ${\rm Div}_1$(Reynolds stress)}}
=
-
\nabla_1\mathsf{p}_1
\label{m1-eqn}
\\
&
\frac{D}{D t_1}\mathbf{m}_2
-
\underbrace{\
\mathbf{u}_2\times{\rm curl}_2\,\mathbf{m}_2\
}_{\hbox{\it Nonlinearity}}
=
-\nabla_2\mathsf{p}_2
\label{m2-eqn}
\end{align}
with ${\rm div}_1\mathbf{u}_1=0$ and  ${\rm div}_2\mathbf{u}_2=0$, by construction. Upon specializing to the $\left\langle\|\mathbf{u}\|^2_{L^2} \right\rangle_2$ averaged Lagrangian \eqref{L2-lagrangian}, one finds $\mathbf{m}_1=\mathbf{u}_1$ and $\mathbf{m}_2=\mathbf{u}_2$ thereby recovering equations \eqref{BigWhirls-mot} and \eqref{SmallWhirls-mot}. In this case, the Reynolds stress term in equation (\ref{m1-eqn}) becomes a gradient $\nabla_1\int \frac12|\mathbf{u}_2|^2d^3x_2 $, which may be absorbed into the pressure gradient.

\begin{remark}[Momentum conservation laws]$\,$\rm

Equations (\ref{m1-eqn})--(\ref{m2-eqn}) imply the following conservation law for the advection of the total fluctuation momentum at a point $\mathbf{x}_1$ by the mean flow,
\begin{align}
\frac{D}{D t_1}\int \mathbf{m}_2(\mathbf{x}_1,\mathbf{x}_2) \,\de^3x_2 = 0
\,.
\label{m2-conserv}
\end{align}
For the kinetic energy in (\ref{L2-lagrangian}) equation (\ref{m2-conserv}) recovers the formula (\ref{mean-preserv}) for the preservation of the mean. In general, this momentum conservation law replaces the preservation of the mean in equation (\ref{mean-preserv}). 

Because $\mathbf{m}_2(\mathbf{x}_1,\mathbf{x}_2):={\delta\ell}/{\delta{\mathbf{u}}_2}$ is a variational derivative of an $\mathbf{x}_1$-translation invariant Lagrangian, it is always possible to differentiate by parts and rewrite the last term in (\ref{M1-eqn}) as the divergence of a stress tensor, 
\[
\int\! \left(\nabla_1\mathbf{u}_2\right)^{\rm T}\cdot \frac{\delta\ell}{\delta{\mathbf{u}}_2}\,\de^3x_2
=
\nabla_1\cdot {\sf R}_1
\,.\]
Consequently, equation (\ref{M1-eqn}) will also conserve the total $\mathbf{x}_1$-momentum,
\[
\frac{d}{dt}\int \frac{\delta\ell}{\delta{\mathbf{u}}_1} \de^3x_1 = 0
\,.\]
Likewise, equation (\ref{m2-conserv}) may be interpreted as a conservation law for the total fluctuation momentum
\[
\frac{d}{dt}\iint \frac{\delta\ell}{\delta{\mathbf{u}}_2}  \de^3x_2\,\de^3x_1 = 0
\,,\]
whose preservation arises from Noether's theorem by $\mathbf{x}_2$-translation invariance. 

\end{remark}

System (\ref{m1-eqn})--(\ref{m2-eqn}) conserves energy and momentum, and is a Lie-Poisson Hamiltonian system, whose Hamiltonian function and Lie-Poisson bracket are given in Appendix \ref{HamStructure}.

\subsection{Evolution of fluctuation labels} \label{full-trans}
In order to compare our results with previous literature, it may be useful to compute the evolution equation for the fluctuation label $\boldsymbol\eta_{\bx_1}(\bx_2^{(0)})$ and compare it with equation \eqref{Adv-microstructure}. This task can be easily accomplished by substituting the Lagrange-to-Euler map:
\begin{equation}
\label{LtoE-map}
\frac{\delta \ell}{\delta \mathbf{u}_2}(\mathbf{x}_1,\mathbf{x}_2,t)=
\iint\!
\frac{\delta \ell}{\delta \mathbf{u}_2}(\mathbf{x}_1^{(0)},\mathbf{x}_2^{(0)},0)\ 
\delta\!\left(\mathbf{x}_2-\boldsymbol\eta_{\boldsymbol\psi(\mathbf{x}_1^{(0)\!\!},\,t)}(\mathbf{x}_2^{(0)\!},t)\right)\,
\delta(\mathbf{x}_1-\boldsymbol\psi(\mathbf{x}_1^{(0)\!},t))\,
\de^3{x}_2^{(0)}\,\de^3{x}_1^{(0)}
\end{equation}
into equation \eqref{M2-eqn}, or equivalently into equation \eqref{m2-eqn}. Upon denoting $\bm_2={\delta \ell}/{\delta \mathbf{u}_2}$ for simplicity, we may pair equation \eqref{m2-eqn} with a divergence-less test function and obtain the following equation of motion
\begin{equation}
\label{KSA-equation}
\frac{\partial}{\partial t}\boldsymbol\eta_{\bx_1}(\bx_2^{(0)\!},t)
+
\underbrace{
\bu_1(\bx_1,t)\cdot\nabla_1\boldsymbol\eta_{\bx_1}(\bx_2^{(0)\!},t)
}
_
\textit{transport along the mean flow}
=
\underbrace{
\bu_2(\boldsymbol\eta_{\bx_1}(\bx_2^{(0)\!},t),\bx_1,t)
}
_\textit{micromotion of fluctuations}
\end{equation}
where we recall $\bu_1=\dot{\boldsymbol\psi}\circ\boldsymbol\psi^{-1}$ and the notation is such that $\bx_1=\boldsymbol\psi(\bx_1^{(0)\!},t)$ and $\boldsymbol\eta_{\bx_1}(\bx_2^{(0)\!},0)=\bx_2^{(0)}$. At this point, the comparison with equation \eqref{Adv-microstructure} is immediate. Indeed, one readily concludes that the important difference between \eqref{Adv-microstructure} and \eqref{KSA-equation} is that the second allows for more freedom in the fluctuation labels, whose micromotion is encoded in the fluctuation velocity term $\bu_2(\boldsymbol\eta_{\bx_1}(\bx_2^{(0)\!},t),\bx_1,t)$ on the right hand side. Dropping this micromotion term yields precisely \eqref{Adv-microstructure}, which assumes that the internal structure 
associated to fluctuations is completely frozen into the mean flow, 
moving with velocity $\bu_1$. Actually, one can say that equation \eqref{KSA-equation} embodies the essence of the present theory by showing explicitly how this differs from previous works in this subject. Notice that the equation \eqref{KSA-equation} is totally equivalent to \eqref{Fluctuations}, although \eqref{KSA-equation} is  more suggestive since the right hand side provides a direct link with the Eulerian velocity $\bu_2$ appearing in the equations of motion \eqref{M1-eqn}--\eqref{M2-eqn}.

\section{Conservation of total SGS circulation and helicity}\label{SGS-circ+hel}

\subsection{Kelvin-Noether circulation theorem}
As shown in the previous section, the evolution of the fluid momentum $\delta\ell/\delta\mathbf{u}_2$ in terms of its initial value is reconstructed by using the Lagrange-to-Euler map \eqref{LtoE-map}. Upon dropping explicit time dependence for convenience, the latter can also be written by pulling back the momentum evolution as follows:
\[
\frac{\delta \ell}{\delta \mathbf{u}_2}(\mathbf{x}_1^{(0)},\mathbf{x}_2^{(0)},0)
=
\iint\!
\frac{\delta \ell}{\delta \mathbf{u}_2}(\mathbf{x}_1,\mathbf{x}_2,t)\ 
\delta\!\left(\mathbf{x}_2^{(0)}-\boldsymbol\eta_{\boldsymbol\psi(\mathbf{x}_1^{(0)})}^{-1}(\mathbf{x}_2)\right)\,
\delta(\mathbf{x}_1^{(0)}-\boldsymbol\psi^{-1}(\mathbf{x}_1))\,
\de^3{x}_2\,\de^3{x}_1
\]
so that \emph{Noether's theorem} reads,
\[
\frac{\de}{\de t} \iint\!
\frac{\delta \ell}{\delta \mathbf{u}_2}(\mathbf{x}_1,\mathbf{x}_2,t)\ 
\delta\!\left(\mathbf{x}_2^{(0)}-\boldsymbol\eta_{\boldsymbol\psi(\mathbf{x}_1^{(0)})}^{-1}(\mathbf{x}_2)\right)\,
\delta(\mathbf{x}_1^{(0)}-\boldsymbol\psi^{-1}(\mathbf{x}_1))\,
\de^3{x}_2\,\de^3{x}_1=0
\,,\]
where expanding the time derivative yields the equation of motion for $\delta\ell/\delta\mathbf{u}_2$.

At this point, one can fix a loop $\gamma_0$ and take the integral
\begin{align*}
&\ \frac{\de}{\de t}\oint_{\gamma_0}\frac{\delta \ell}{\delta \mathbf{u}_2}(\mathbf{x}_1^{(0)},\mathbf{x}_2^{(0)},0)\cdot\de{\mathbf{x}_2^{(0)}}
\\
=&\ 
\frac{\de}{\de t}\oint_{\gamma_0\!}\left(\iint\!
\frac{\delta \ell}{\delta \mathbf{u}_2}(\mathbf{x}_1,\mathbf{x}_2,t)\ 
\delta\!\left(\mathbf{x}_2^{(0)}-\boldsymbol\eta_{\boldsymbol\psi(\mathbf{x}_1^{(0)})}^{-1}(\mathbf{x}_2)\right)\,
\delta(\mathbf{x}_1^{(0)}-\boldsymbol\psi^{-1}(\mathbf{x}_1))\,
\de^3{x}_2\,\de^3{x}_1\right)\!\cdot\de{\mathbf{x}_2^{(0)}}
\end{align*}
to obtain the following Kelvin-Noether theorem.
\begin{theorem}[Kelvin fluctuation-circulation density]
At each point $\mathbf{x}_1$, consider a loop $\gamma(\mathbf{u}_2)$ moving with the fluctuation velocity $\mathbf{u}_2(\mathbf{x}_1,\mathbf{x}_2)={\dot{\boldsymbol\eta}_{\mathbf{x}_1}\circ\boldsymbol\eta^{-1}_{\mathbf{x}_1}({\mathbf{x}_2})}$. Then, the following transport dynamics holds:
\begin{align}
\frac{D\mathcal{K}}{D t_1}=0
\quad\hbox{where}\quad
\mathcal{K}(\mathbf{x}_1,t) := \oint_{\gamma(\mathbf{u}_2)}
\frac{\delta \ell}{\delta \mathbf{u}_2}(\mathbf{x}_1,\mathbf{x}_2)\cdot\de\mathbf{x}_2
\,.
\label{KNthm-eqn}
\end{align}
\end{theorem}
The proof can be found in Appendix \ref{KNthm-proof}. 
\begin{remark}\rm
Spatial integration over $\mathbf{x}_1$ in (\ref{KNthm-eqn}) yields conservation of the total fluctuation-circulation
\begin{align}
\frac{\de}{\de t}\int\mathcal{K}\, \de^3x_1=0
\,,
\label{conserv-fluctcirc}
\end{align}
upon using $\nabla_1\cdot\mathbf{u}_1=0$.
\end{remark}

\subsection{Vorticity and helicity density}
Applying the corresponding curl operations to equations (\ref{m1-eqn}) and (\ref{m2-eqn}) yields,
upon defining 
\[
\boldsymbol{\omega}_1 = \nabla_1\times\mathbf{m}_1
\quad\hbox{and}\quad
\boldsymbol{\omega}_2 = \nabla_2\times\mathbf{m}_2
\,,\]
the following vorticity dynamics
\begin{align}
&\frac{\partial}{\partial t}\boldsymbol{\omega}_1
- \nabla_1\times (\mathbf{u}_1\times\boldsymbol{\omega}_1)
+
\nabla_1\times
\int\! \left(\nabla_1\mathbf{u}_2\right)^{\rm T}\cdot \mathbf{m}_2\,\de^3x_2
=
0
\label{omega1-eqn}
\\
&
\frac{D}{D t_1}\boldsymbol{\omega}_2
- \nabla_2\times (\mathbf{u}_2\times\boldsymbol{\omega}_2)
= 0
\,.\label{omega2-eqn}
\end{align}
Consequently, the \emph{fluctuation helicity density}
\begin{equation}
\mathcal{H}(\mathbf{x}_1,t)=\int\!\mathbf{m}_2\cdot\operatorname{curl}\mathbf{m}_2\,\de^3x_2
\label{H-def}
\end{equation}
satisfies the transport equation
\begin{equation}
\frac{D\mathcal{H}}{D t_1}=0
\,.
\label{H-advect}
\end{equation}
This calculation implies the following.

\begin{theorem}\rm
The total fluctuation helicity is preserved,
\[
\frac{d}{dt}\int \mathcal{H}\,\de^3x_1=0\,.
\]
\end{theorem}

\begin{proof}
The proof proceeds as a direct calculation,
\[
\frac{d}{dt}\int \mathcal{H}\,\de^3x_1
=
\int \frac{\partial}{\partial t} \mathcal{H}\,\de^3x_1 
=
- \int \nabla_1(\mathcal{H}\mathbf{u}_1)\,\de^3x_1 
=0
\]
for periodic boundary conditions.
\end{proof}

\section{Resolved-scale models based on SGS advection}

\subsection{A resolved-scale model that is structurally similar to MHD}
The advection laws (\ref{KNthm-eqn}) and (\ref{H-advect}) for the total SGS circulation $\mathcal{K}$ and helicity $\mathcal{H}$, respectively, suggest replacing the SGS degrees of freedom with these Lagrangian-averaged quantities obtained from $\mathbf{x}_2$-integration.
Hamilton's principle then takes the form
\[
\delta\int_{t_1}^{t_2}\ell(\mathbf{u}_1,\mathcal{H},\mathcal{K})\,\de t=0
\,.
\]
This variational principle results in  the following equations of motion for $\mathbf{m}_1 := {\delta \ell}/{\delta \mathbf{u}_1}$, $\mathcal{H}$ and $\mathcal{K}$,
\begin{align}
&\frac{D\mathbf{m}_1}{D t_1}
+
 \left(\nabla_1\mathbf{u}_1\right)^{\rm T}\cdot\mathbf{m}_1
+
\frac{\delta \ell}{\delta \mathcal{H}}\nabla_1 \mathcal{H}
+
\frac{\delta \ell}{\delta \mathcal{K}}\nabla_1 \mathcal{K}
=
-
\nabla_1\mathsf{p}_1
\label{BigWhirls-motHK}
\\
&
\frac{D\mathcal{H}}{D t_1}=0
\,,\qquad
\frac{D\mathcal{K}}{D t_1}=0
\,.
\label{SmallWhirls-motHK}
\end{align}
An illustrative special case can be found by choosing the Lagrangian
\begin{align}
\ell(\mathbf{u}_1,\mathcal{H},\mathcal{K})
=
\frac12\int\left|\mathbf{u}_1(\mathbf{x}_1)\right|^2
-
\frac12 |\nabla_1 \mathcal{H}\times\nabla_1 \mathcal{K}|^2
\,\de^3{x}_1
\,.
\label{MHD-Lagrangian}
\end{align}
Upon defining $\mathbf{B}:=\nabla_1 \mathcal{H}\times\nabla_1 \mathcal{K}$, the corresponding equations (\ref{BigWhirls-motHK}) and (\ref{SmallWhirls-motHK}) arising from Hamilton's principle for this Lagrangian are
\begin{align}
&\frac{D\mathbf{u}_1}{D t_1}
+  \mathbf{B} \times {\rm curl}_1 \mathbf{B}
=
-
\nabla_1\mathsf{p}_1
\,,
\qquad \nabla_1\cdot \mathbf{u}_1 = 0
\,,
\label{MHD-mot}\\
&
\frac{\partial \mathbf{B}}{\partial t}
=
{\rm curl}_1( \mathbf{u}_1 \times \mathbf{B})
\,,
\qquad \nabla_1\cdot \mathbf{B} = 0
\,.
\label{MHD-B}
\end{align}
These are \emph{identical} to the equations of ideal incompressible magnetohydrodynamics (MHD). 
However, in these equations, the quantity $\mathbf{B}$ is the cross product of the $\mathbf{x}_1$-gradients of the $\mathbf{x}_2$-integrated SGS  helicity $\mathcal{H}$ and circulation densities $\mathcal{K}$, defined in equations (\ref{H-def}) and (\ref{KNthm-eqn}), respectively. The equations in (\ref{MHD-B}) imply the relation 
\begin{align}
\frac{D\mathbf{B}}{Dt_1} = \mathbf{B}\cdot\nabla_1\bu_1
\,,\quad\hbox{so}\quad
\Big[\frac{D}{Dt_1}\,,\,\mathbf{B}\cdot\nabla_1\Big]=0
\,,
\label{MHD-B2}
\end{align}
where $[\,\cdot\,,\,\cdot\,]$ is the commutator of divergence-free vector fields. Consequently, one finds the relation
\begin{align}
 \frac{D^2\mathbf{B}}{Dt_1^2}
 = \frac{D}{Dt_1}(\mathbf{B}\cdot\nabla_1 \bu_1) 
 = -\,\mathbf{B}\cdot\nabla_1\Big(\nabla_1\left(\mathsf{p}_1 + \tfrac12|\mathbf{B}|^2\right) 
 - \mathbf{B}\cdot\nabla_1\mathbf{B}\Big)
\,.
\label{MHD-B3}
\end{align}

\begin{remark}\rm
The two formulas (\ref{MHD-B2}) and (\ref{MHD-B3}) correspond to Ertel's theorem \cite{Er1942} and Ohkitani's relation \cite{Oh1993} for Euler's equations with vorticity $\boldsymbol{\omega}=\nabla\times \bu$, respectively,
\begin{align}
\Big[\frac{D}{Dt}\,,\,\boldsymbol{\omega}\cdot\nabla\Big]=0
\quad\hbox{and}\quad
 \frac{D^2\boldsymbol{\omega}}{Dt^2}
 = \frac{D}{Dt}(\boldsymbol{\omega}\cdot\nabla \bu) 
 = -\,\boldsymbol{\omega}\cdot\nabla\nabla\mathsf{p}
\,.
\label{MHD-EO}
\end{align}
\end{remark}

\begin{remark}\rm
Equation (\ref{MHD-B3}) may be rearranged into the form of a \emph{nonlinear wave equation}
\begin{align}
 \frac{D^2\mathbf{B}}{Dt_1^2}
 -
\Big( \mathbf{B}\cdot\nabla_1\Big)^2\mathbf{B}
 = -\,\mathbf{B}\cdot\nabla_1\nabla_1\left(\mathsf{p}_1 + \tfrac12|\mathbf{B}|^2\right) 
\,.
\label{MHD-B4}
\end{align}
We have shown that the SGS advection laws (\ref{KNthm-eqn}) and (\ref{H-advect}) for the integrated SGS circulation $\mathcal{K}$ and helicity  $\mathcal{H}$ impart a certain elasticity to the resolved scale equations, which is completely analogous to Alfv\'en waves in MHD. The magnitude $|\mathbf{B}|$ is not preserved by this flow. In fact, equation (\ref{MHD-B2}) implies 
\begin{align}
\frac12\frac{D|\mathbf{B}|^2}{Dt_1} = \mathbf{B}\cdot\left(\nabla_1\mathbf{u}_1\right)\cdot\mathbf{B}
\label{MHD-magB}
\end{align}
for the evolution of the magnitude. 
\end{remark}

\begin{remark}\rm
If the Lagrangian in (\ref{MHD-Lagrangian}) were actually relevant in turbulence, the result (\ref{MHD-B4}) would have vast implications for the spatiotemporal properties of turbulence. In particular, the cross product of resolved gradients $\nabla_1\mathcal{K}$ and $\nabla_1\mathcal{H}$ of the integrated SGS circulation and helicity would propagate as a wave. This is a degree of freedom not usually considered in turbulence models. However, we have treated the Lagrangian (\ref{MHD-Lagrangian}) here just for illustration. We do not expect it to actually describe turbulence. We have discussed it only to show that an energetic dependence on the resolved gradients of the integrated advected properties of the subgrid scales could lead to interesting dynamical behaviour in the resolved scales, by imparting a type of potential energy that would lead to a type of non-Newtonian behaviour. The concept of non-Newtonian properties of turbulence has a long history of investigation, going back at least to \cite{Ri1957}. More recently, the related concept of nonlinear dispersion in turbulence has also been studied intensely, see \cite{FoHoTi2001,FoHoTi2002}. However, as far as we are aware, the closest analog in the previous literature of this concept of elasticity of integrated SGS properties is due to \cite{MP1997}, who derived similar equations from a different approach, based on stretching of SGS vortices.
A derivation of a class of stretched-vortex SGS models similar to those discussed in \cite{MP1997} will be provided using a variational principle in the next section. 
\end{remark}

\begin{remark}\rm
The two formulas (\ref{MHD-B2}) and (\ref{MHD-B3}) correspond to Ertel's theorem and Ohkitani's relation \cite{Oh1993} for Euler's equations with vorticity $\boldsymbol{\omega}=\nabla\times \bu$, respectively,
\begin{align}
\Big[\frac{D}{Dt}\,,\,\boldsymbol{\omega}\cdot\nabla\Big]=0
\quad\hbox{and}\quad
 \frac{D^2\boldsymbol{\omega}}{Dt^2}
 = \frac{D}{Dt}(\boldsymbol{\omega}\cdot\nabla \bu) 
 = -\,\boldsymbol{\omega}\cdot\nabla\nabla\mathsf{p}
\,.
\label{MHD-EO}
\end{align}
\end{remark}

\subsection{Convection in the stretched SGS vortex model 2 of \cite{MP1997}}

In this section, we introduce the constrained kinetic energy Lagrangian
\begin{align}
\ell 
=
\int
\frac{D}{2}|\bu_1|^2 - p(D-1) - q\, (|\mathbf{B}|^2 - 1 )\,d\,^3x_1
\,,\label{MP-Lag}
\end{align}
with $\mathbf{B}:=\nabla_1 \mathcal{H}\times\nabla_1 \mathcal{K}$ and $Dd\,^3x_1$ the preserved volume element.  The Lagrange multiplier $p$ (the pressure) imposes volume preservation. The Lagrange multiplier $q$ imposes $|\mathbf{B}|^2 = 1$, so that, as in \cite{MP1997}, we may think of the subgrid scale order parameter as a vortex filament whose strength is constant, but whose direction varies. 
The Euler-Poincar\'e motion equation for this Lagrangian is obtained as in \cite{HMR1998} and found to be
\begin{align}
\frac{D\bu_1}{Dt_1}
=
-\nabla_1 p - q \nabla_1|\mathbf{B}|^2 
-
\underbrace{\
 \nabla_1\cdot 2q\,\left({\rm Id} - \frac{\mathbf{B}\otimes\mathbf{B}}{|\mathbf{B}|^2}\right)
}_{\hbox{MP97 SGS stress form}}
.\label{MP-EPeqns}
\end{align}
The motion equation (\ref{MP-EPeqns}) \emph{conserves} the kinetic energy $\frac12\int |\bu_1|^2\,d\,^3x_1$. This means that, unlike most SGS models of turbulence, the subgrid scale vortices in the present model do not dissipate resolved-scale kinetic energy. Instead, all dissipation of kinetic energy at the resolved scales must be modelled separately; for example, as a viscous term that could be added later at the $\mathbf{x}_1$-scale. Models of dissipation are not part of our discussion here. 

The last expression in equation (\ref{MP-EPeqns}) has the same form as the stress term in equation (25) of \cite{MP1997}, although the corresponding expression there does not transform properly as a tensor. Instead, the expression for the stress tensor there transforms under change of variables as an array of scalars. This means it neglects the convection of the subgrid vortices by the resolved field, as discussed explicitly in  \cite{MP1997}. It also means that the stretched SGS vortex model 2 of \cite{MP1997} is not variational. 
The present variational model also contains auxiliary equations for preservation of volume and advection of $\mathbf{B}$ by the resolved field, with constant magnitude,
\begin{align}
\frac{\partial D}{\partial t}
=
-\,{\rm div}_1\,(D\bu_1)
\,,\qquad
\frac{\partial \mathbf{B}} {\partial  t}
=
{\rm curl}_1\,(\bu_1\times\mathbf{B})
\,, \qquad  |\mathbf{B}|^2=1
\,.
\label{MP-aux}
\end{align}
The Lagrange multipliers $p$ and $q$ obey a
system of linear equations found by preservation
of initial conditions $D=1$ and $|\mathbf{B}|^2 = 1$, which require, respectively
\begin{align}
-
\frac{1}{D}\frac{D }{Dt_1}D
=
\nabla_1\cdot\bu_1=0
\quad\hbox{and}\quad
\frac12\frac{D|\mathbf{B}|^2}{Dt_1} 
=
\mathbf{B}\cdot\left(\nabla_1\mathbf{u}_1\right)\cdot\mathbf{B}= 0
\quad\hbox{with}\quad
\nabla_1\cdot\mathbf{B}=0
\,.\label{MP-constraints}
\end{align}
\paragraph{Summary}
\begin{enumerate}[(i)]
\item
When $q$ is interpreted spectrally as $2q(t,\bx_1)=:\int_{k_c(t,\bx_1)}^\infty E(k)dk$ for a cut-off wavenumber $k_c(t,\bx_1)$, equations (\ref{MP-EPeqns})--(\ref{MP-constraints}) represent a variant of the MP97 model 2 in which we have restored the convection of the subgrid vortices by
the resolved field that was neglected in \cite{MP1997}.
\item
When $2q\equiv1$, this model reduces to ideal incompressible MHD, in which case ${D|\bB|^2}/{Dt_1}\ne0$, by equation (\ref{MHD-magB}). 
\end{enumerate}

\paragraph{Ertel Theorem and Ohkitani relation for the present variant of MP97 model 2.}
The present variant of the MP97 model 2 motion equation (\ref{MP-EPeqns}) may be written as 
\begin{align}
\frac{D\bu_1}{Dt_1}
=
-\nabla_1 (p+2q) + \nabla_1\cdot 2q\,\mathbf{B}\otimes\mathbf{B}
=:
\bF
\,,
\label{MP97-mot}
\end{align}
where preservation of $\nabla_1\cdot\bu_1=0$ and $\mathbf{B}\cdot \nabla_1\bu_1 \cdot \mathbf{B} = 0$
determines $p$ and $q$, and
\begin{align}
\frac{\partial \mathbf{B}} {\partial  t}
=
{\rm curl}_1\,(\bu_1\times\mathbf{B})
\,, \quad  |\mathbf{B}|^2=1
\quad\hbox{with}\quad
\nabla_1\cdot\mathbf{B}=0
\,.\label{MP97-mag}
\end{align}
The Ertel Theorem and Ohkitani relations for this variant of MP97 model 2 are then
\begin{align}
\Big[\frac{D}{Dt_1}\,,\,\mathbf{B}\cdot\nabla_1\Big]=0
\,,\hbox{ for }
\frac{D\mathbf{B}}{Dt_1} = \mathbf{B}\cdot\nabla_1\bu_1
\hbox{ and }
 \frac{D^2\mathbf{B}}{Dt_1^2}
 = \frac{D}{Dt_1}(\mathbf{B}\cdot\nabla_1 \bu_1) 
 = \mathbf{B}\cdot\nabla_1\bF
\label{MP97-EO}
\end{align}
Together, the Ertel and Ohkitani relations conveniently deliver
\pone
\[
\frac{D}{Dt_1}(\mathbf{B}\cdot\nabla_1 \bu_1 \cdot\mathbf{B})
=
\mathbf{B}\cdot\nabla_1\bF\cdot\mathbf{B}
+
|\mathbf{B}\cdot\nabla_1\bu_1|^2
\]
which must vanish in order to preserve the constraint that $\mathbf{B}^2=1$.

\paragraph{The equation system for Lagrange multipliers $p$ and $q$.}
Preservation of $\nabla_1\cdot\bu_1=0$ and $\mathbf{B}\cdot \nabla_1\bu_1 \cdot
\mathbf{B} = 0$ determines the Lagrange multipliers $p$ and $q$ from the following system
\begin{align}
\begin{split}
&\frac{\partial \,}{\partial t}(\nabla_1\cdot\bu_1)
=
-\,|\nabla_1\bu_1|^2 + \nabla_1\cdot\bF =0
\,,\\
&\frac{D}{Dt_1}(\mathbf{B}\cdot\nabla_1 \bu_1 \cdot\mathbf{B})
=
\mathbf{B}\cdot\nabla_1\bF\cdot\mathbf{B}
+
|\mathbf{B}\cdot\nabla_1\bu_1|^2 =0
\,,\\
&\hbox{in which}\quad
\bF
:=
-\nabla_1 (p+2q) + \nabla_1\cdot (2q\,\mathbf{B}\otimes\mathbf{B})
\,.\end{split}
\label{qp-system}
\end{align}
That is, the force $\bF$ for this variational version of MP97 model 2 depends linearly on  $p$ and $q$ as in (\ref{MP97-mot}). An additional $q$-term enters the boundary conditions for this system. Namely, the normal component of the force $\bF$ must vanish on a fixed, flat boundary; that is, $\bF\cdot\mathbf{\hat{n}}=0$.

\begin{remark}\rm
See the papers \cite{MP1997,VPC2000} for discussions of numerical implementations of the stretched-vortex subgrid-stress model, as well as discussions of its applications to forced and decaying turbulence, and studies of its numerical complexity relative to standard Large Eddy Simulation (LES) models such as the Smagarinsky model. The present variational version of the stretched-vortex subgrid-stress model restores convection of the subgrid vortices by the resolved field. However, it has the additional complexity that the system of equations (\ref{qp-system}) for $p$ and $q$ must also be solved at each time step. The solvability of this system has not been studied yet and such a study would be beyond the scope of the present work. 
\end{remark}

\section{SGS fluid model with feedback between the scales}

\subsection{Energy coupling of fluctuations and stress tensor}

In this section, we shall return to the multiscale description and consider the particular Lagrangian
\begin{align}
\ell(\mathbf{u}_1,\mathbf{u}_2)
&=
\frac12\int\left|\mathbf{u}_1(\mathbf{x}_1)\right|^2 \de^3 \mathbf{x}_1
+
\frac12\iint \Big(\left|{\mathbf{u}}_2\big(\mathbf{x}_1,\mathbf{x}_2\big)\right|^2
+\hspace{-2mm}
\underbrace{\
\alpha_1^2\left|\nabla_1{\mathbf{u}}_2\right|^2
}_{\hbox{Coupling term}}\hspace{-2mm}
\Big)
\de^3\mathbf{x}_2\,\de^3 \mathbf{x}_1
\,,\label{BasicLag}
\end{align}
where $\alpha_1$ is an appropriately chosen \emph{coupling constant} with the dimensions of \emph{length}.
In this case, the mean-fluctuation interaction energy is the $L^2$ norm of the $\mathbf{x}_1$-gradient of the slowly varying envelope of the fluctuation velocity $\mathbf{u}_2(\mathbf{x}_1,\mathbf{x}_2)$. 
The momenta $\mathbf{m}_1$ and $\mathbf{m}_2$ in equations (\ref{m1-eqn})--(\ref{m2-eqn}) are given in terms of the velocities $\mathbf{u}_1$ and $\mathbf{u}_2$ for this Lagrangian as
\[\mathbf{m}_1=\mathbf{u}_1
\quad\hbox{and}\quad
\mathbf{m}_2=(1-\alpha_1^2\Delta_1)\mathbf{u}_2
\,,\quad\hbox{hence}\quad
\mathbf{u}_2 = (1-\alpha_1^2\Delta_1)^{-1}\mathbf{m}_2
\,,\]
where $\Delta_1$ is the Laplacian operator in the $\mathbf{x}_1$ coordinates, so that $(1-\alpha_1^2\Delta_1)^{-1}$ is a smoothing operator. 
Consequently, the coupling term in (\ref{BasicLag}) introduces a \emph{Reynolds stress} term that will remain and affect the evolution, as follows.
\begin{align}
&\frac{D}{D t_1}\mathbf{u}_1
-
\underbrace{\
\alpha_1^2 \nabla_1\cdot
\int\! \left(\nabla_1\mathbf{u}_2^{\rm T}\cdot \nabla_1\mathbf{u}_2
- {\rm Id}|\nabla_1\mathbf{u}_2|^2\right) \,\de^3x_2\
}
_{\hbox{\it (Divergence of Reynolds stress, ${\rm div}_1 {\sf R}_1$)}}
=
-
\nabla_1\mathsf{p}_1
\label{newm1-eqn}
\\
&
\frac{D}{D t_1}(\mathbf{u}_2-\alpha_1^2\Delta_1\mathbf{u}_2)
-
\underbrace{\
\mathbf{u}_2\times{\rm curl}_2\,(\mathbf{u}_2-\alpha_1^2\Delta_1\mathbf{u}_2)\
}_{\hbox{\it (Nonlinear convection)}}
=
-\nabla_2\mathsf{p}_2
\,.
\label{newm2-eqn}
\end{align}
Upon introducing the subgrid Reynolds stress notation ${\sf R}_2$, the latter equation becomes
\begin{align}
&\frac{D}{D t_1}(\mathbf{u}_2-\alpha_1^2\Delta_1\mathbf{u}_2)
+
\hspace{-4mm}
\underbrace{\
\nabla_2\cdot 
\Big({\sf R}_2 - {\rm Id}\,\mathsf{p}_2\Big)
}_{\hbox{\it (Subgrid Reynolds stress)}}
\hspace{-4mm}
=
0
\label{newm2-eqn-div}
\end{align}
where the term 
\begin{align}
 -\, \mathbf{u}_2\times{\rm curl}_2\,(\mathbf{u}_2-\alpha_1^2\Delta_1\mathbf{u}_2)
= \nabla_2\cdot  {\sf R}_2-\mathbf{u}_2 (\nabla_2\cdot \mathbf{u}_2)
\label{R2-stresstensor}
\end{align}
has zero mean for $\nabla_2\cdot\mathbf{u}_2=0$ and the transformation to $\nabla_2\cdot  {\sf R}_2$ stress-divergence form in (\ref{newm2-eqn-div}) represents momentum conservation, which arises from Noether's theorem for the $\mathbf{x}_2$-translation invariant Lagrangian in (\ref{BasicLag}).

\begin{enumerate}
\item
{\bf Reynolds stress.}
The only channel of feedback from small scales to the larger ones arises from the \emph{Reynolds stress} term in equation (\ref{newm1-eqn}).  The Reynolds stress tensor 
\begin{align}
{\sf R}_1=\alpha^2_1\int\!\nabla_1\mathbf{u}_2^{\rm T}\cdot \nabla_1\mathbf{u}_2 \,\de^3x_2
\quad\hbox{with components}\quad
{\sf R}_1^{ij}
=\alpha^2_1\int\! \frac{\partial u_2^i}{\partial x_1^k}\,\frac{\partial u_2^j}{\partial x_1^k} \,\de^3x_2
\label{Reynolds-stress}
\end{align}
in equation (\ref{newm1-eqn}) is remarkably similar to the tensor diffusivity stress tensor for large-eddy simulation (LES) introduced  in \cite{Le1974} and discussed in detail in \cite{WWVJ2001}. It exchanges energy between the smaller scales and the larger scales along the $L^2$ mean primary stretching direction of $\nabla_1\mathbf{u}_2$, and the direction of this exchange depends on the gradient of the slowly varying envelope of the fluctuation velocity.
There is, however, one very important difference between the Reynolds stress in (\ref{Reynolds-stress}) and the corresponding term in an LES model: our model does not yet take viscous effects into account. Indeed the system (\ref{newm1-eqn})--(\ref{newm2-eqn}) conserves the kinetic energy expressed by the Lagrangian in (\ref{BasicLag}), and even comprises a Lie-Poisson Hamiltonian system, as discussed in Appendix \ref{HamStructure}.

The Reynolds stress term in our system represents backscatter due to nonlinear dispersion, not diffusion. Thus, as with the Euler-alpha model \cite{HMR1998,FoHoTi2002} the higher order terms that make the solutions more regular do so without introducing dissipation.  

\item
{\bf SGS circulation and helicity advection.} Note that the present model retains the same advection laws (\ref{KNthm-eqn}) and (\ref{H-advect}) for the total SGS circulation and helicity as in Section \ref{SGS-circ+hel}.  Therefore, the corresponding resolved-scales exist and may be studied for this choice of the multiscale Lagrangian, as well. 

\item
{\bf Taylor hypothesis.}
The Taylor hypothesis for the small scales assumes that \emph{their nonlinearity does not affect their evolution}; so that they are passively swept by the larger scales. Applying the Taylor hypothesis would neglect the term marked \emph{Nonlinear convection} in equation (\ref{newm2-eqn}), or \emph{Subgrid Reynolds stress} in equation (\ref{newm2-eqn-div}). Even if it eliminated the nonlinear convection effects of the fluctuation velocity, the Taylor hypothesis would not eliminate the effects of the back-reaction of the small scales due to their sweeping by the large scales, which still remain in the term labelled as \emph{Reynolds stress} in equation  (\ref{newm1-eqn}).  

\item
{\bf Regularity of the equations.}
Equation (\ref{newm2-eqn}) is reminiscent of the Navier-Stokes-alpha model \cite{FoHoTi2002}, but with smoothing applied to the envelope of the fluctuations in a frame moving with the mean flow. In fact, equation (\ref{newm2-eqn}) is essentially a two-scale version of the Euler-alpha model \cite{HMR1998} in a frame moving with the $\mathbf{x}_1$-flow.

The Lagrangian (\ref{BasicLag}) could be modified further to allow introduction of norms that would be strong enough to ensure long-time existence of its corresponding solutions, even in the absence of viscosity. In the presence of viscosity, the resulting equations are:
\begin{align}
&\frac{D}{D t_1}\mathbf{u}_1
+
{\rm div}_1 {\sf R}_1
=
-
\nabla_1\mathsf{p}_1
+
\nu_1 \Delta_1 \mathbf{u}_1
\,,\label{nu-m1-eqn}
\\
&
\frac{D}{D t_1}(\mathbf{u}_2-\alpha_1^2\Delta_1\mathbf{u}_2)
+
{\rm div}_2 {\sf R}_2
=
-\nabla_2\mathsf{p}_2
+
\nu_2 \Delta_2
(\mathbf{u}_2-\alpha_1^2\Delta_1\mathbf{u}_2)
\,,
\label{nu-m2-eqn}
\end{align}
with  ${\sf R}_2$ and ${\sf R}_1$ given in (\ref{R2-stresstensor}) and (\ref{Reynolds-stress}), respectively.
The flows at both scales are incompressible, so the pressures $p_1$ and $p_2$ in equations (\ref{nu-m1-eqn}) and (\ref{nu-m2-eqn}) are determined from preservation of ${\rm div}_1\mathbf{u}_1=0$ and ${\rm div}_2\mathbf{u}_2=0$, respectively.
The viscosity has been introduced \emph{ad hoc} here, as the diffusion of momentum at the each scale.
Equations (\ref{nu-m1-eqn})--(\ref{nu-m2-eqn}) may admit global strong solutions, as occurs for the Navier-Stokes-alpha model investigated in \cite{FoHoTi2002}. However, their analysis is beyond the scope of the present paper and will be left to the future. 

\item
{\bf Numerical complexity and implementation of multiscale models.}
It is possible that one may be able to reduce the computational complexity of these equations and thereby accelerate their  computation by introducing a type of optimal sampling that would select those smaller scales that make the largest contributions to the Reynolds stress term. 
The multiscale analysis in the present case couples the resolved scale solution with a subgrid cell problem for vortex filament evolution at each point of the resolved scale grid. The computational cost for this coupled system of equations could be quite expensive, although there are some alternatives that might be used to lessen the cost. For example, an adaptive scheme recently has been developed to reduce complexity and speed up the computation in a related case. See \cite{HYR2008} for a discussion of a promising approach for reducing the complexity of such multiscale computations.  See also \cite{EfendHou2009} for an approach to multiscale computations using Finite Element methods.  Yet another possibility might be to adapt the heterogeneous multiscale method \cite{E-etal} that has already been developed for complex fluids to the present case. The difference is that in the present case the subgrid scale dynamics governs an infinite-dimensional vortex filament, rather than a finite-dimensional order parameter, as occurs in liquid crystals and magnetic fluids. 
In any case, the numerical simulations of the equations for the present multiscale model and the development of numerical algorithms for their solution must be left to the future. 
\end{enumerate}

\section{Conclusions}

Using a simple decomposition argument that applies to all periodic functions, the present approach has sought to transfer ideas associated with convection of microstructure in complex fluids into the context of turbulence modeling. The microstructure was interpreted as an order parameter in the same fashion as spin is regarded for ferromagnetic fluids. This interpretation was suggested by previous work in which turbulence microstructure had been assumed to be transported by the mean flow. However, in this paper, the intrinsic nonlinear features of microstructure were also considered so that microstructure underwent its own nonlinear evolution in the frame of the mean flow. In short, the \emph{kinematic sweeping ansatz} (KSA) assumed that the mean flow serves as a Lagrangian frame of motion for the  fluctuation dynamics. This seems to be an effective approximation, in general, and it was contrasted in section \ref{full-trans} with the exact formula. Other assumptions about convection of microstructure would have been possible. However, the KSA has the advantage of possessing a purely Eulerian description, and thereby avoiding the necessity of computing the Lagrangian back-to-labels map. In addition, the geometric framework that was previously developed for complex fluids could be naturally transferred to turbulence, thereby leading, in the present case, to pure transport dynamics for the fluctuation-circulation and the fluctuation-helicity densities.

The resulting model arising from KSA has also the advantage that additional, even finer, scales could also be incorporated by simple extension of the geometric features. In the resulting iterated hierarchy of smaller and smaller scales, the flow associated to a certain scale serves as a Lagrangian coordinate for the flow of the next finer scale, and so on. 

Even in the simplest 2-scale version treated here, the KSA yields dynamical equations that we hope will be suitable for modelling purposes. For example, when the lengthscale $\alpha$ is inserted in Hamilton's principle for smoothness requirements, the resulting model produces Reynold's stress tensor in a way that is quite reminiscent of the diffusivity stress tensor in LES simulations.

Various features of the present model still remain to be addressed. For example, vortex dynamics at the smaller scales could be an interesting future direction for research, in both 2D and 3D. Indeed, although the $\mathbf{u}_1$ equation may lose the vortex filament solution, the $\mathbf{u}_2$ has a vortex filament solution at each Eulerian point $\mathbf{x}_1$, see Appendix \ref{App-vortex}. Hence, projecting onto the plane would yield a proliferation of point vortices at each point in physical space. In addition, the development of numerical integrators that respect the geometric framework here would be a very important advance, both in new mathematics and toward evaluating the present model in applications. For discussions of recent progress in this direction see \cite{GBPav2012}. Finally, we hope that the methods of \cite{E-etal} might also be profitably applied to the present multiscale model of ideal fluid motion. 

\subsection{Acknowledgments}
We thank our friends P. Constantin, C. J. Cotter, C. R. Doering, F. Gay-Balmaz, J. D. Gibbon, J. Pietarila Graham, T. S. Ratiu and B. Wingate for their kind encouragement and thoughtful remarks during the course of this work. DDH gratefully acknowledges partial support by the Royal Society of London's Wolfson Award scheme and the European Research Council's Advanced Grant. 

\appendix

\section{Appendix}

\subsection{Proof of Lemma 1}\label{lem1-proof}

Consider the Fourier series expansion of a periodic function $\mathbf{u}(\mathbf{x})$  on the unit cube $[0, 1]^3$,
\[
\mathbf{u}(\mathbf{x}) = \sum_{\mathbf{k}\in\mathbb{Z}^3}\mathbf{\widehat{u}}(\mathbf{k}) e^{{} i\mathbf{k} \cdot \mathbf{x}}
,\quad
  \mathbf{k}\in\mathbb{Z}^3,
\]
in which $\mathbf{\widehat{u}}(\mathbf{k})$ are the Fourier coefficients. Let $0 \le 1/N < 1$ to be a reference
wavelength, where $N$ is an integer, and one denotes
\[
\Lambda_N = \{\mathbf{k}\in\mathbb{Z}^3; |\mathbf{k}| \le N\}
\quad\hbox{and}\quad
 \Lambda_N' = \mathbb{Z}^3 \backslash \Lambda_N
.\]
Decompose the function $\mathbf{u}$ into two additive parts as follows:
\[
\mathbf{u}(\mathbf{x}) = \mathbf{u}_{1}(\mathbf{x}) + \mathbf{u}_{2}(\mathbf{x}),
\]
where
\begin{align*}
\mathbf{u}_{1}(\mathbf{x}) &= \sum_{\mathbf{k}\in\Lambda_N} \mathbf{\widehat{u}}(\mathbf{k}) e^{{} i\mathbf{k} \cdot \mathbf{x}},
\quad\hbox{and}\quad
\mathbf{u}_{2}(\mathbf{x}) = \sum_{\mathbf{k}\in\Lambda_N'}\mathbf{\widehat{u}}(\mathbf{k}) e^{{} i\mathbf{k} \cdot \mathbf{x}}.
\end{align*}
Rewrite $\mathbf{k} = \mathbf{k}_{1} + N\mathbf{k}_{2}$ where $\mathbf{k}_{1}$ and $\mathbf{k}_{2}$ take integer values in $\mathbb{Z}^3$ with $\mathbf{k}_{1}\in\Lambda_N$, and compute
\begin{align*}
\mathbf{u}_{2} 
 &= \sum_{\mathbf{k}\in\Lambda_N'} \mathbf{\widehat{u}}(\mathbf{k}) e^{{} i\mathbf{k} \cdot \mathbf{x}}
\\&
=
\sum_{(\mathbf{k}_{1} + N\mathbf{k}_{2})\in\Lambda_N'} \mathbf{\widehat{u}}(\mathbf{k}_{1} + N\mathbf{k}_{2}) 
e^{{} i(\mathbf{k}_{1} + N\mathbf{k}_{2})\cdot \mathbf{x}}
\\&
=
\sum_{\mathbf{k}_{2}\ne0}
\left(
 \sum_{\mathbf{k}_{1}\in\Lambda_N} \mathbf{\widehat{u}}(\mathbf{k}_{1} + N\mathbf{k}_{2}) e^{i\mathbf{k}_{1} \cdot \mathbf{x}}
 \right)
 e^{i\mathbf{k}_{2} \cdot N\mathbf{x}}
\\&
 =:
\sum_{\mathbf{k}_{2}\ne0}
\mathbf{\widehat{u}}_{2}(\mathbf{k}_{2}, \mathbf{x})  e^{i\mathbf{k}_{2} \cdot N\mathbf{x}}
 \\&
 =: \mathbf{u}_{2}(\mathbf{x},\mathbf{x}/\epsilon)
\end{align*}
where $\epsilon= 1/N$ and the quantity
\[
\mathbf{\widehat{u}}_{2}(\mathbf{k}_{2}, \mathbf{x}) := \sum_{\mathbf{k}_{1}\in \Lambda_N}
\mathbf{\widehat{u}}(\mathbf{k}_{1} + N\mathbf{k}_{2}) e^{i\mathbf{k}_{1} \cdot \mathbf{x}}
\,,\]
involves only Fourier modes whose wave number is less than $N$ in magnitude. Hence, any periodic function $\mathbf{u}\in\mathbb{R}^3$ may be rewritten as
\[
\mathbf{u}(\mathbf{x}) = \mathbf{u}_{1}(\mathbf{x}) + \mathbf{u}_{2}(\mathbf{x},\mathbf{x}/\epsilon) = \mathbf{u}_1(\mathbf{x}_1) + \mathbf{u}_2(\mathbf{x}_1,\mathbf{x}_2)
\]
where $\mathbf{u}_{2}(\mathbf{x}_1,\mathbf{x}_2)$ is a periodic function in $\mathbf{x}_2$ with mean zero, 
\[
\int  \mathbf{u}_2(\mathbf{x}_1,\mathbf{x}_2) \,d\,^3x_2 = 0\,,
\]
since $\mathbf{k}_{2}\ne0$. Note that the functions $\mathbf{u}_{1}$ and $\mathbf{u}_{2}$ depend on the choice of $\epsilon=1/N$. $\blacksquare$

\begin{remark}
The decomposition
\begin{align*}
\mathbf{u}(\mathbf{x}) &= \mathbf{u}_{1}(\mathbf{x}) + \mathbf{u}_{2}(\mathbf{x},\mathbf{x}/\epsilon)
\\
&= \sum_{\mathbf{k}\in\Lambda_N} \mathbf{\widehat{u}}(\mathbf{k}) e^{{} i\mathbf{k} \cdot \mathbf{x}}
+ 
\sum_{\mathbf{k}_{2}\ne0}
\mathbf{\widehat{u}}_{2}(\mathbf{k}_{2}, \mathbf{x})  e^{i\mathbf{k}_{2} \cdot \mathbf{x}/\epsilon}
\end{align*}
may be regarded as a Fourier-series generalization of the WKB form 
\[
\mathbf{u}(\mathbf{x},\mathbf{x}/\epsilon) = \mathbf{\overline{u}}(\mathbf{x}) 
+
\frac12 \left( \mathbf{a}(\mathbf{x})e^{i\theta(\mathbf{x})/\epsilon}
+  \mathbf{a}^*(\mathbf{x})e^{-i\theta(\mathbf{x})/\epsilon}\right)
\]
This WKB form was used in \cite{GjHo1996} to develop a wave, mean flow interaction theory by applying Lagrangian averaging in Hamilton's principle for rotating, stratified incompressible flow.  The present work is similar in spirit to that previous work. 
\end{remark}

\subsection{Expressions for the velocity variations}\label{velocity-vars}

The Euler-Poincar\'e equations are obtained from Hamilton's principle with 
\begin{align*}
\delta \mathbf{u}_1 &= \dot{\xi}_1 - {\rm ad}_{u_1}\xi_1
\\
\delta \mathbf{u}_2 &= \dot{\xi}_2 - {\rm ad}_{u_2}\xi_2 
+ \pounds_{\mathbf{u}_1}\xi_2   - \pounds_{\xi_1}\mathbf{u}_2
\\
&= \dot{\xi}_2 - {\rm ad}_{u_2}\xi_2  
+ \left(\mathbf{u}_1\cdot\nabla_1\right)\xi_2
- \left({\xi_1}\cdot\nabla_1\right)\mathbf{u}_2 
\end{align*}
where ${\rm div}_1\mathbf{u}_1=0$,  ${\rm div}_2\mathbf{u}_2=0$, $\pounds_\mathbf{w}\mathbf{v}=\mathbf{w}\cdot\nabla_1\mathbf{v}$ denotes the Lie derivative,  and
\[
{\rm ad}_{u_2}\xi_2  =u_2\cdot\nabla_2\xi_2 -\xi_2 \cdot\nabla_2 u_2
\,,
\]

In particular, while the first variation follows directly by taking $\delta\mathbf{u}_1=\delta(\dot{\boldsymbol\psi}\boldsymbol\psi^{-1})$ and by defining $\xi_1=\delta{\boldsymbol\psi}\boldsymbol\psi^{-1}$, the second variation follows by the calculation below, which uses the pullback notation:
\begin{align*}
\delta\mathbf{u}_2
&\ =
\delta\left(\boldsymbol\psi^*\left(\left(\frac{\de}{\de t}\left(\boldsymbol\psi^*\boldsymbol\eta\right)\right)\left(\boldsymbol\psi^*\boldsymbol\eta\right)^{-1}\right)\right)
\\
&\ =
\pounds_{\xi_1}{u_2}
+
\boldsymbol\psi^*\left(\left(\frac{\de}{\de t}\!\left(\delta\!\left(\boldsymbol\psi^*\boldsymbol\eta\right)\right)\right)\left(\boldsymbol\psi^*\boldsymbol\eta\right)^{-1}\right)
-
\boldsymbol\psi^*\left(\left(\frac{\de}{\de t}\left(\boldsymbol\psi^*\boldsymbol\eta\right)\right)\left(\boldsymbol\psi^*\boldsymbol\eta\right)^{-1}\delta\!\left(\boldsymbol\psi^*\boldsymbol\eta\right)\left(\boldsymbol\psi^*\boldsymbol\eta\right)^{-1}\right)
\\
&\ =
\pounds_{\xi_1}{u_2}
-
{u_2}\cdot\nabla_2{\xi_2}
+
\dot{\xi_2}
-\pounds_{u_1}\xi_2
+
{\xi_2}\cdot\nabla_2{u_2}
\end{align*}
where we have denoted $\xi_2=\boldsymbol\psi^*\!\left(\left(\delta\!\left(\boldsymbol\psi^*\boldsymbol\eta\right)\right)\left(\boldsymbol\psi^*\boldsymbol\eta\right)^{-1}\right)$.

\subsection{Proof of the Kelvin-Noether theorem}\label{KNthm-proof}
Upon denoting
\[
\left(\frac{\delta \ell}{\delta \mathbf{u}_2}\right)_{\!t}=\frac{\delta \ell}{\delta \mathbf{u}_2}(\mathbf{x}_1,\mathbf{x}_2,t)
\,,\qquad
\left(\frac{\delta \ell}{\delta \mathbf{u}_2}\right)_{\!0}=\frac{\delta \ell}{\delta \mathbf{u}_2}(\mathbf{x}_1^{(0)},\mathbf{x}_2^{(0)},0)
\,,
\] 
and by using the pullback notation, we have
\[
\left(\frac{\delta \ell}{\delta \mathbf{u}_2}\right)_{\!t}=\boldsymbol\psi_*\left(\boldsymbol\psi^*\boldsymbol\eta\right)_*\left(\frac{\delta \ell}{\delta \mathbf{u}_2}\right)_{\!0}
\ \Longrightarrow\ 
\left(\frac{\delta \ell}{\delta \mathbf{u}_2}\right)_{\!0}=\left(\boldsymbol\psi^*\boldsymbol\eta\right)^*\boldsymbol\psi^*\left(\frac{\delta \ell}{\delta \mathbf{u}_2}\right)_{\!t}
\]
Thus, we can take the circulation around a fixed loop $\gamma_0$ in $\Bbb{R}^3_2$
\begin{align*}
0=&\
\frac{\de}{\de t}\oint_{\gamma_0}\left(\boldsymbol\psi^*\boldsymbol\eta\right)^*\boldsymbol\psi^*\left(\frac{\delta \ell}{\delta \mathbf{u}_2}\right)_{\!t}
\\&\
=
\frac{\de}{\de t}\oint_{\widetilde{\gamma}}\boldsymbol\psi^*\left(\frac{\delta \ell}{\delta \mathbf{u}_2}\right)_{\!t}
\\&\
=
\frac{\de}{\de t}\!\left(\boldsymbol\psi^*\!\oint_{\gamma}\left(\frac{\delta \ell}{\delta \mathbf{u}_2}\right)_{\!t\,}\right)
\\&\
=
\boldsymbol\psi^*\!\left(\frac{\partial}{\partial t}\oint_{\gamma}\frac{\delta \ell}{\delta \mathbf{u}_2}\cdot\de\mathbf{x}_2+({\mathbf{u}_1\cdot\nabla_1})\oint_{\gamma}\frac{\delta \ell}{\delta \mathbf{u}_2}\cdot\de\mathbf{x}_2\right)
\end{align*}
where 
\[
\widetilde{\gamma}=\left(\boldsymbol\psi_*\boldsymbol\eta\right)\circ\gamma_0=\boldsymbol\eta_{\boldsymbol\psi(\mathbf{x}_1)}(\gamma_0)
\]
 moves with velocity $\widetilde{\mathbf{u}}_2$ and 
 \[
 \gamma=\boldsymbol\psi_*\widetilde{\gamma}=\boldsymbol\eta_{\boldsymbol\psi^{-1}\circ\boldsymbol\psi(\mathbf{x}_1)}=\boldsymbol\eta_{\mathbf{x}_1}(\gamma_0)
\] 
moves with velocity $\mathbf{u}_2=\boldsymbol\psi^*\widetilde{\mathbf{u}}_2$. Thus, upon applying $\boldsymbol\psi_*$,  we have the Kelvin circulation theorem
\[
\frac{\partial}{\partial t}\oint_{\gamma({\mathbf{u}}_2)}\frac{\delta \ell}{\delta \mathbf{u}_2}\cdot\de\mathbf{x}_2+({\mathbf{u}_1\cdot\nabla_1})\oint_{\gamma({\mathbf{u}}_2)}\frac{\delta \ell}{\delta \mathbf{u}_2}\cdot\de\mathbf{x}_2=0
\,.\ \blacksquare
\]
\bigskip 

\subsection{Hamiltonian structure}\label{HamStructure}

Upon defining $\mathbf{m}_i=\delta \ell/\delta \mathbf{u}_i$, one writes the functional Legendre transformation
\[
h(\mathbf{m}_1,\mathbf{m}_2)=\int\!\mathbf{m}_1\cdot\mathbf{u}_1\,\de^3 x_1+\iint\!\mathbf{m}_2\cdot\mathbf{u}_2\,\de^3 x_2\,\de^3 x_1-\ell(\mathbf{u}_1,\mathbf{u}_2)
\,,
\]
so that the equations of motion \eqref{M1-eqn} and \eqref{M2-eqn} read as
\begin{align}
&\frac{\partial\mathbf{m}_1}{\partial t}
+
\frac{\delta h}{\delta{\mathbf{m}}_1}\cdot\nabla_1\mathbf{m}_1
+
\left(\nabla_1\frac{\delta h}{\delta{\mathbf{m}}_1}\right)\cdot\mathbf{m}_1
+
\int\!\left(\nabla_1\frac{\delta h}{\delta{\mathbf{u}}_2}\right)^{\!\!\rm T\!}\cdot \mathbf{m}_2\,\de^3x_2
=
-
\nabla_1\mathsf{p}_1
\label{M1-Ham-eqn}
\\
&\frac{\partial\mathbf{m}_2}{\partial t}
+
\frac{\partial}{\partial x_1^j} \left(\frac{\delta h}{\delta{u}_1^j}\, \mathbf{m}_2\right)
+
\frac{\delta h}{\delta{\mathbf{m}}_2}\cdot\nabla_2\mathbf{m}_2
+
\left(\nabla_2\frac{\delta h}{\delta{\mathbf{m}}_2}\right)^{\!T\!}\cdot\mathbf{m}_2
=
-\nabla_2\mathsf{p}_2
\,.
\label{M2-Ham-eqn}
\end{align}
Therefore, by standard methods \cite{HMR1998}, the Poisson bracket for the above system is
\begin{multline*}
\{f,h\}=\int\!\mathbf{m}_1\cdot\left[\frac{\delta f}{\delta{\mathbf{m}}_1},\frac{\delta h}{\delta{\mathbf{m}}_1}\right]_1\de^3x_1+\iint\!\mathbf{m}_2\cdot\left[\frac{\delta f}{\delta{\mathbf{m}}_2},\frac{\delta h}{\delta{\mathbf{m}}_2}\right]_2\de^3x_2\,\de^3x_1
\\
+\iint{\mathbf{m}_2}
\cdot
\left(
\frac{\delta h}{\delta{\mathbf{m}}_1}\cdot\nabla_1\frac{\delta f}{\delta{\mathbf{m}}_2}
-
\frac{\delta f}{\delta{\mathbf{m}}_1}\cdot\nabla_1\frac{\delta h}{\delta{\mathbf{m}}_2}
\right)
\de^3x_2\,\de^3x_1
\end{multline*}
where we have denoted the Lie algebra brackets as
\[
\left[\mathbf{v},\mathbf{w}\right]_i=\mathbf{w}\cdot\nabla_i\mathbf{v}-\mathbf{v}\cdot\nabla_i\mathbf{w}
\,,
\]
for $i=1,2$.

\subsection{Vortex structures in the subgrid scales}\label{App-vortex}

This appendix shows how equations \eqref{newm1-eqn}-\eqref{newm2-eqn} allow for vortex structures in the fluctuation dynamics. In order to see how this happens, we can write \eqref{m1-eqn}-\eqref{m2-eqn} in terms of the vorticities
\[
\boldsymbol\omega_1=\nabla_1\times\bm_1
\,,
\qquad
\boldsymbol\omega_2=\nabla_2\times\bm_2
\]
which yields \eqref{omega1-eqn}-\eqref{omega2-eqn} with
\begin{align}\label{diamond}
\int\! \left(\nabla_1\mathbf{u}_2\right)^{\rm T}\cdot \mathbf{m}_2\,\de^3x_2
=
-\int\! \left(\nabla_1\boldsymbol\omega_2\right)^{\rm T}\cdot\boldsymbol\phi_2 \,\de^3x_2
\end{align}
where $\boldsymbol\phi_2$ is the vector potential associated to $\bu_2$, so that $\bu_2=\nabla_2\times\boldsymbol\phi_2$.

Notice that equation \eqref{omega2-eqn} possesses a vortex solution of the type
\[
\boldsymbol\omega_2(\bx_1,\bx_2,t)=\int\!\partial_s\mathbf{R}(\bx_1,s,t)\,\delta(\bx_2-\mathbf{R}(\bx_1,s,t))\,\de s
\,,
\]
so that a vortex filament in the fluctuation vorticity is attached to each point of the mean fluid.

In turn, the above vortex solution can be used to reduce the level of difficulty of the equations \eqref{m1-eqn}-\eqref{m2-eqn}. Indeed, replacing the above solution into \eqref{diamond} eliminates the integral over $x_2$ thereby yielding
\[
-\int\! \left(\nabla_1\boldsymbol\omega_2\right)^{\rm T}\cdot\boldsymbol\phi_2 \,\de^3x_2
=
-\left.\int\!\left(\boldsymbol\phi_2 \cdot\nabla_1\mathbf{R}'
-\mathbf{R}'\,\nabla_2\boldsymbol\phi_2:\nabla_1\mathbf{R}\right)
\,\de s\right|_{\bx_2=\mathbf{R}(\bx_1,s,t)}
\]
where $\mathbf{R}'=\partial_s\mathbf{R}$ and $A:B=A_{ij}B_{ij}$. The above term shows how the subgrid vortex determines the strain tensor of the mean flow.
On the other hand, the dynamics of the vortex filament can be derived by pairing equation \eqref{omega2-eqn} with a test vector field, thereby yielding
\[
(\partial_t+\bu_1\cdot\nabla_1)\mathbf{R}
=\left.\nabla_2\times\boldsymbol\phi_2 \right|_{\bx_2=\mathbf{R}(\bx_1,s,t)}
= \mathbf{u}_2(\bx_1, \mathbf{R}(\bx_1,s,t),t)
\,,
\]
which shows how the mean velocity affects the subgrid vortex via the material time derivative.

\end{document}